\documentclass[a4paper,11pt]{article} 

\usepackage{fullpage}
\usepackage{times}
\usepackage{soul}
\usepackage{url}
\usepackage[utf8]{inputenc}
\usepackage[small]{caption}
\usepackage{graphicx}
\usepackage{xcolor}
\usepackage{amsmath}
\usepackage{mathtools}
\mathtoolsset{showonlyrefs = true}
\usepackage{booktabs}
\usepackage{algorithm}
\usepackage{enumerate}
\usepackage{dsfont}
\usepackage{cleveref}
\usepackage{comment}

\usepackage{amsthm}
\usepackage{amssymb}

\newtheorem{theorem}{Theorem}
\newtheorem{corollary}{Corollary}
\newtheorem{proposition}{Proposition}

\newtheorem{claim}{Claim}

\theoremstyle{definition}
\newtheorem{definition}{Definition}
\newtheorem{example}{Example}

\usepackage{natbib}
\usepackage{authblk}

\allowdisplaybreaks

\title{\bf Decomposition Envy-Freeness in Random Assignment}

\author[1]{Yasushi Kawase}
\author[2]{Warut Suksompong}
\author[3]{Hanna Sumita}
\author[3]{Yu Yokoi}

\affil[1]{Chuo University, Japan}
\affil[2]{National University of Singapore, Singapore}
\affil[3]{Institute of Science Tokyo, Japan}

\date{\vspace{-10mm}}

\newenvironment{claimproof}{\par
  \pushQED{$\hfill \lhd$}%
  \normalfont \topsep 6pt\relax
  \trivlist
  \item[\hskip\labelsep
        {\textit{Proof (Claim).}}]\ignorespaces
}{%
  \popQED\endtrivlist
}

\begin{document}

\maketitle

\begin{abstract}
In random assignment, fairness is often captured by \emph{stochastic-dominance envy-freeness (SD-EF)}. We observe that assignments satisfying SD-EF may admit decompositions that result in each agent envying another agent with high probability.
To address this, we introduce \emph{decomposition envy-freeness (Dec-EF)}, which is a property of a decomposition rather than of an assignment matrix. We show that an SD-EF assignment matrix always admits a Dec-EF decomposition when there are at most three agents or the agents have at most two distinct preferences.
\end{abstract}

\section{Introduction}
\label{sec:intro}

The assignment problem concerns allocating indivisible objects to agents based on their preferences.
Example applications include allocating school seats to students, public houses to tenants, and job positions to workers.
Since it is often impossible to be completely fair to the agents---for example, even if all agents share the same preference over objects, only one agent can get the best object, one agent can get the second-best object, and so on---randomness is frequently used in order to promote fairness.

In the classic setting of random assignment popularized by \citet{BogomolnaiaMo01}, there is the same number of agents and objects, and each agent submits a strict ordinal preference over the objects.
The goal is to compute a random assignment, which is a probability distribution over deterministic assignments, satisfying desirable properties.
\citeauthor{BogomolnaiaMo01} examined two assignment rules: \emph{random priority (RP)} and \emph{probabilistic serial (PS)}.
They showed that the output of PS always satisfies a fairness notion defined based on the \emph{stochastic dominance (SD)} relation called \emph{SD-envy-freeness (SD-EF)}.
This notion requires that for any pair of agents $i, i'$ and any object~$o$, the probability that agent~$i$ receives an object that is at least as preferable as~$o$ (according to the preference of agent~$i$) is no less than the corresponding probability of agent~$i'$ (still according to the preference of agent~$i$).
On the other hand, the output of RP sometimes fails SD-EF, but always satisfies a weaker version of it.
Other properties in random assignment, such as efficiency and strategyproofness, have also been studied by \citeauthor{BogomolnaiaMo01} along with numerous subsequent authors \citep{McLennan02,AbdulkadirogluSo03,KattaSe06,KojimaMa10,MennleSe21,Duddy25}.

Observe that the definition of SD-EF depends only on the ``assignment matrix'', which indicates the probability that each agent receives each object, and not on the distribution over deterministic assignments.
For example, PS outputs an assignment matrix and does not naturally suggest which distribution over deterministic assignments to use.
The focus on assignment matrices is typically justified by the Birkhoff--von Neumann theorem, which states that every bistochastic matrix (i.e., a square matrix with nonnegative entries such that the entries in each row or column of the matrix sum to~$1$) can be written as a convex combination of permutation matrices; note that an assignment matrix is always bistochastic.
Therefore, each assignment matrix can be decomposed into a probability distribution over deterministic assignments.
This fact led \citet{BogomolnaiaMo01} to write that ``two probability distributions over [deterministic assignments] resulting in the same bistochastic matrix will not be distinguished''.
Although the probability distribution over deterministic assignments is generally not unique, to the best of our knowledge, other studies of random assignment have also ignored the choice of which distribution to implement.\footnote{Some studies of allocation problems have considered ex-ante and ex-post fairness \citep{BudishChKo13,AzizFrSh24}.}

Despite the overwhelming focus on assignment matrices in the literature, we observe that the choice of distribution over deterministic assignments can play a crucial role as far as envy-freeness is concerned on an intuitive level.
This is illustrated by the following example.
\begin{example}
\label{ex:cyclic}
Consider an instance with three agents, where every agent has the preference $a \succ b\succ c$ over the three objects.
Clearly, the uniform assignment matrix that gives each object to each agent with probability~$1/3$ satisfies SD-EF.
However, consider the decomposition that produces the assignment $(1{:}a, 2{:}b, 3{:}c)$ with probability~$1/3$, the assignment $(1{:}b, 2{:}c, 3{:}a)$ with probability~$1/3$, and the assignment $(1{:}c, 2{:}a, 3{:}b)$ with probability~$1/3$.
Then, the probability that agent~$2$ envies agent~$1$ in the chosen deterministic assignment is $2/3$.
By contrast, for the decomposition that puts probability~$1/6$ on each of the six deterministic assignments, any agent envies any other agent with probability only $1/2$.
\end{example}
With more agents and objects, a direct generalization of Example~\ref{ex:cyclic} (with the ``cyclic'' decomposition) implies that in a decomposition of an SD-EF assignment matrix, it is possible that one agent envies another agent with probability arbitrarily close to~$1$.
This undermines the intended fairness of the allocation process, e.g., if agents are aware of the chosen decomposition, or if the process is repeated several times and an agent notices that she envies a particular agent almost every time.\footnote{\citet{Duddy25} offered a different critique of SD-EF by arguing that a random assignment satisfying it can be unfair from an egalitarian perspective.}

To address this issue, we propose a new envy-freeness notion, \emph{decomposition envy-freeness (Dec-EF)}, which is a property of a distribution over deterministic assignments (rather than of an assignment matrix).
Dec-EF demands that for any pair of agents $i, i'$, the probability that $i$ envies $i'$ in a deterministic assignment drawn from the distribution is at most $1/2$.
In Example~\ref{ex:cyclic}, the cyclic decomposition fails Dec-EF, since each agent envies some other agent with probability $2/3$.
Note that this decomposition is also output by the so-called Birkhoff's algorithm as well as other heuristics, which are often aimed at minimizing the number of deterministic assignments in the decomposition \citep{DufosseUc16,DufosseKaPa18}.
On the other hand, the decomposition that puts probability~$1/6$ on each of the six deterministic assignments, which is intuitively fairer, satisfies Dec-EF.

In light of this, a natural question is whether an SD-EF assignment matrix always admits a Dec-EF decomposition---if so, this would mean that for any SD-EF assignment, there is a ``fair'' decomposition that can be used.
We answer this question in the affirmative for two cases: when there are at most three agents, and when the agents have at most two distinct preferences.

\section{Preliminaries}

There is a set~$N$ of $n$~agents and a set~$M$ of $n$~objects, for some positive integer $n\ge 2$.
Each agent $i\in N$ has a strict preference $\succ_i$ over the objects in~$M$.
An \emph{instance} consists of the agents, objects, and agents' preferences.
A \emph{deterministic assignment} is a one-to-one matching between agents and objects.
A \emph{random assignment} is a probability distribution over deterministic assignments.
Given a random assignment, the corresponding \emph{assignment matrix} is an $n\times n$ matrix that indicates the probability with which each agent receives each object.
Note that the assignment matrix is always bistochastic, i.e., the entries are nonnegative and the entries in each row or column sum to~$1$.
By the Birkhoff--von Neumann theorem, for every bistochastic matrix, there exists a random assignment with that matrix as its assignment matrix---we call such a random assignment a \emph{decomposition} of the assignment matrix.
Given an assignment matrix, its decomposition is generally not unique.

A \emph{random allocation} is a probability distribution over the $n$~objects.
Thus, an assignment matrix~$P$ can be viewed as a collection of $n$ random allocations, one for each agent; we use $P_i$ to denote agent~$i$'s random allocation in~$P$.
For an agent $i\in N$, a random allocation $X$ is said to \emph{stochastically dominate} another random allocation~$Y$ with respect to $\succ_i$ if for each object~$o\in M$, the total probability of objects at least as preferable as~$o$ in~$X$ according to~$\succ_i$ is no less than the corresponding probability in~$Y$.
In other words, $\sum_{o':\, o' \succeq_i o} X_{o'} \geq \sum_{o':\, o'\succeq_i o} Y_{o'}$ holds for each object $o$, where $X_{o'}$ and $Y_{o'}$ denote the probabilities of the object $o'$ in $X$ and $Y$, respectively.
For example, if agent~$i$ has the preference $o_1\succ_i o_2\succ_i o_3$, then the random allocation that gives probabilities $(0.8, 0.1, 0.1)$ to $o_1,o_2,o_3$ in this order stochastically dominates the one that gives probabilities $(0.5, 0.3, 0.2)$, but does not stochastically dominate the one that gives probabilities $(0, 1, 0)$.
\begin{definition}
An assignment matrix~$P$ satisfies:
\begin{itemize}
\item \emph{SD-envy-freeness (SD-EF)} if for each pair of agents $i,i'\in N$, the random allocation~$P_i$ of agent~$i$ stochastically dominates the random allocation~$P_{i'}$ of agent~$i'$ with respect to~$\succ_i$;
\item \emph{weak SD-envy-freeness (weak SD-EF)} if for each pair of agents $i,i'\in N$, either $P_i = P_{i'}$ or $P_{i'}$ does not stochastically dominate~$P_i$ with respect to $\succ_i$.
\end{itemize}
\end{definition}

We now introduce our notion of envy-freeness.
Recall that a random assignment refers to a probability distribution over deterministic assignments rather than an assignment matrix.
\begin{definition}
\label{def:Dec-EF}
A random assignment~$\pi$ satisfies \emph{decomposition envy-freeness (Dec-EF)} if for every pair of agents $i,i'\in N$, the probability that agent~$i$ envies agent~$i'$ in a deterministic assignment drawn from~$\pi$ is at most $1/2$.
\end{definition}
Note that $1/2$ is the smallest reasonable number that can be chosen in this definition.
For example, suppose all agents have the same preference.
Then, for any random assignment and any pair of agents, the probability that some agent envies the other agent must be at least $1/2$.
This means that if a number smaller than $1/2$ is chosen in \Cref{def:Dec-EF}, no random assignment can satisfy the definition.

As discussed in \Cref{sec:intro}, it is possible that for an assignment matrix, one decomposition satisfies Dec-EF while another one does not.

\begin{definition}
An assignment matrix satisfies \emph{envy-freeness decomposability (EF-decomposability)} if it admits a decomposition that satisfies Dec-EF.
\end{definition}

The following example shows that a weak SD-EF assignment matrix may not be EF-decomposable.
\begin{example}
\label{ex:weak-SD-EF-Dec-EF}
Consider an instance with $n = 3$ agents whose preferences are $a\succ_1 b\succ_1 c$,\, $a\succ_2 c\succ_2 b$, and $c\succ_3 a\succ_3 b$.
Let $P$ be the assignment matrix
\[ P = 
\begin{bmatrix}
0 & 1 & 0 \\
0.9 & 0 & 0.1 \\
0.1 & 0 & 0.9
\end{bmatrix},
\]
where the rows correspond to agents $1,2,3$ and the columns correspond to objects $a,b,c$, respectively.
One can check that $P$ satisfies weak SD-EF.\footnote{In contrast, $P$ does not satisfy SD-EF, since $P_1$ stochastically dominates neither $P_2$ nor $P_3$ with respect to $\succ_1$.}
However, the only decomposition of $P$ into deterministic assignments is
\[ \pi = 0.9 \cdot 
\begin{bmatrix}
0 & 1 & 0 \\
1 & 0 & 0 \\
0 & 0 & 1
\end{bmatrix} 
+ 0.1 \cdot 
\begin{bmatrix}
0 & 1 & 0 \\
0 & 0 & 1 \\
1 & 0 & 0
\end{bmatrix}.
\]
This decomposition does not satisfy Dec-EF, as agent~$1$ envies agent~$2$ with probability~$0.9$.
\end{example}

Conversely, a Dec-EF random assignment may not correspond to a weak SD-EF assignment matrix either.

\begin{example}
\label{ex:Dec-EF-weak-SD-EF}
Consider an instance with $n = 3$ agents whose preferences are $a\succ_1 b\succ_1 c$,\, $b\succ_2 c\succ_2 a$, and $c\succ_3 a\succ_3 b$.    
The random assignment
\[ \pi = 0.5 \cdot 
\begin{bmatrix}
1 & 0 & 0 \\
0 & 1 & 0 \\
0 & 0 & 1
\end{bmatrix} 
+ 0.5 \cdot 
\begin{bmatrix}
0 & 0 & 1 \\
1 & 0 & 0 \\
0 & 1 & 0
\end{bmatrix}
\]
satisfies Dec-EF, where the rows correspond to agents $1,2,3$ and the columns correspond to objects $a,b,c$, respectively.
However, the assignment matrix corresponding to $\pi$ is
\[  
P = \begin{bmatrix}
0.5 & 0 & 0.5 \\
0.5 & 0.5 & 0 \\
0 & 0.5 & 0.5
\end{bmatrix}.
\]
This assignment matrix fails weak SD-EF, since $P_2$ stochastically dominates $P_1$ with respect to $\succ_1$ and $P_1\ne P_2$.
\end{example}

\Cref{ex:Dec-EF-weak-SD-EF} illustrates that Dec-EF is a relatively weak notion on its own: under the random assignment~$\pi$, agent~$1$ has large envy toward agent~$2$ half of the time, and is only marginally happier in the other half.
However, as \Cref{ex:cyclic} demonstrates, Dec-EF can sometimes help strengthen the fairness provided by SD-EF.

Two important random assignment rules are  \emph{probabilistic serial (PS)} and \emph{random priority (RP)} \citep{BogomolnaiaMo01}.
On the one hand, an assignment matrix produced by PS always satisfies SD-EF; however, the description of PS does not suggest how to decompose this matrix into deterministic assignments.\footnote{Choosing an arbitrary decomposition can result in a random assignment that violates Dec-EF.
For example, given the instance in Example~\ref{ex:cyclic}, PS returns the uniform assignment matrix, and we have seen that some decomposition of this matrix fails Dec-EF.}
On the other hand, RP naturally yields a random assignment based on its description.
Specifically, RP works by choosing a uniformly random permutation of the agents and then, following this order, implementing a \emph{serial dictatorship} in which the agents, one by one, pick their most preferred object among the remaining ones. Thus, each order results in a deterministic assignment.
Even though the assignment matrix output by RP may violate SD-EF, we observe that the resulting random assignment always satisfies Dec-EF.

\begin{proposition}
\label{prop:RP}
For every instance, the random assignment output by RP satisfies Dec-EF.
\end{proposition}

\begin{proof}
Consider any pair of agents $i,i'\in N$.
In half of the $n!$ agent permutations, $i$ appears before $i'$.
For each such permutation, $i$ picks an object before~$i'$, so $i$ does not envy~$i'$ in the resulting deterministic assignment.
Therefore, across all permutations, the probability that $i$ envies $i'$ is at most $1/2$.
It follows that Dec-EF is satisfied.
\end{proof}

\section{Main Results}

In this section, we prove our main results, which establish the relation between SD-EF and Dec-EF in two cases.
The first case is when there are at most three agents.

\begin{theorem}
\label{thm:three-agents}
For every instance with $n \le 3$, every SD-EF assignment matrix is EF-decomposable.
\end{theorem}

\begin{proof}
The case $n = 2$ will be covered by \Cref{thm:two-types}, so we assume that $n = 3$.
Consider any instance with three agents, and let $P$ be an SD-EF assignment matrix.
We will show that $P$ can be decomposed into a random assignment that satisfies Dec-EF.

Let $N=\{1,2,3\}$, $M=\{o_1,o_2,o_3\}$, and let $p_{ij}$ denote the probability that agent $i$ receives object $o_j$, i.e.,
\begin{align} 
P &= 
\begin{bmatrix}
p_{11} & p_{12} & p_{13} \\
p_{21} & p_{22} & p_{23} \\
p_{31} & p_{32} & p_{33}
\end{bmatrix}.
\end{align}
Let $p^*=\min_{i\in N}\min_{o_j\in M}p_{ij}$.
We decompose this matrix by first subtracting~$p^*$ from every entry of the matrix using the uniform decomposition where each of the six deterministic assignments gets probability $p^*/2$, and then decomposing the remaining matrix (which must contain at least one $0$) in the unique way.\footnote{To see the uniqueness, assume for example that $p_{12} = 0$. Then, $p_{21}$ can be used in only the deterministic assignment involving $p_{13}$, $p_{21}$, $p_{32}$. Similarly, each of $p_{31}$, $p_{23}$, and $p_{33}$ can be used in only one deterministic assignment. When $p_{21},p_{31},p_{23},p_{33}$ are all $0$, there cannot be three nonzero values in different rows and columns, so the entire matrix must be~$0$.}
We show that this decomposition satisfies Dec-EF.

Without loss of generality, we may assume that agent $1$'s preference is $o_1\succ_1 o_2\succ_1 o_3$.
Since $P$ is an SD-EF assignment matrix, we have
$p_{11}\ge \max\{p_{21},\,p_{31},\,1/3\}$ and $p_{11}+p_{12}\ge \max\{p_{21}+p_{22},\,p_{31}+p_{32},\,2/3\}$.

It suffices to show that, under the above decomposition, the probability that agent $1$ envies each of the other agents is at most~$1/2$.
We perform a case analysis based on the position of the minimum entry.
For $a,b,c$ such that $\{a,b,c\} = \{1,2,3\}$, we denote by $A^{abc}$ the permutation matrix corresponding to the deterministic assignment where agents $1,2,3$ get $o_a, o_b, o_c$, respectively.

\begin{itemize}
    \item If $p_{11}=p^*$, then since $P$ is SD-EF, we have $1=p_{11}+p_{21}+p_{31}\le 3p_{11}=3p^*$.
    As $p^*\ge 1/3$ is the smallest value of the entries and $P$ is bistochastic, every entry of~$P$ must be $1/3$.
    Hence, the decomposition is the uniform decomposition over all six deterministic assignments. This implies that agent $1$ envies each of the other agents with probability at most $1/2$.

    \item If $p_{12}=p^*$, then the decomposition becomes
    \begin{align}
    (p_{33}-\tfrac{p^*}{2})A^{123}+
    (p_{23}-\tfrac{p^*}{2})A^{132}+
    \tfrac{p^*}{2}A^{213}+
    \tfrac{p^*}{2}A^{231}+
    (p_{21}-\tfrac{p^*}{2})A^{312}+
    (p_{31}-\tfrac{p^*}{2})A^{321}.
    \end{align}
    We can confirm that this is indeed $P$ by a simple calculation: for example, the $(1,1)$ entry is $(p_{33}+p_{23})-p^* = (p_{11}+p_{12})-p^* = p_{11}$ because $p_{11}+p_{12} = 1 - p_{13} = p_{23}+p_{33}$.
    In this decomposition, agent $1$ envies agent $2$ when one of $A^{213}, A^{312}, A^{321}$ is chosen. Thus, agent $1$ envies agent $2$ with probability
    \begin{align}
        \tfrac{p^*}{2}+\left(p_{21}-\tfrac{p^*}{2}\right)+\left(p_{31}-\tfrac{p^*}{2}\right)
        &=p_{21}+p_{31}-\tfrac{p^*}{2} \\
        &=(1-p_{11})-\tfrac{p_{12}}{2}\\
        &= 1-\tfrac{p_{11}}{2}-\tfrac{p_{11}+p_{12}}{2}
        \le 1-\tfrac{1/3}{2}-\tfrac{2/3}{2}=\tfrac{1}{2}.
    \end{align}
    The probability that agent $1$ envies agent $3$ is also at most $1/2$ by the same argument.

    \item If $p_{13}=p^*$, then the decomposition becomes
    \begin{align}
    (p_{22}-\tfrac{p^*}{2})A^{123}+
    (p_{32}-\tfrac{p^*}{2})A^{132}+
    (p_{21}-\tfrac{p^*}{2})A^{213}+
    (p_{31}-\tfrac{p^*}{2})A^{231}+
    \tfrac{p^*}{2}A^{312}+
    \tfrac{p^*}{2}A^{321}.
    \end{align}   
    In this decomposition, agent $1$ envies agent $2$ with probability
    \begin{align}
        \left(p_{21}-\tfrac{p^*}{2}\right)+\tfrac{p^*}{2}+\tfrac{p^*}{2}
        &=p_{21}+\tfrac{p^*}{2}
        \le \tfrac{p_{21}+p_{11}}{2}+\tfrac{p^*}{2}
        =\tfrac{p_{32}+p_{33}}{2}+\tfrac{p^*}{2}
        \le \tfrac{p_{31}+p_{32}+p_{33}}{2}=\tfrac{1}{2},
    \end{align}
    where the second equality holds by $p_{21}+p_{11} = 1 - p_{31} = p_{32}+p_{33}$.
    The probability that agent~$1$ envies agent $3$ is also at most $1/2$ by the same argument.

    \item If $p_{21}=p^*$, then the decomposition becomes
    \begin{align}
    (p_{33}-\tfrac{p^*}{2})A^{123}+
    (p_{32}-\tfrac{p^*}{2})A^{132}+
    \tfrac{p^*}{2}A^{213}+
    (p_{12}-\tfrac{p^*}{2})A^{231}+
    \tfrac{p^*}{2}A^{312}+
    (p_{13}-\tfrac{p^*}{2})A^{321}.
    \end{align}
    In this decomposition, agent $1$ envies agent $2$ with probability
    \begin{align}
        \tfrac{p^*}{2}+\tfrac{p^*}{2}+\left(p_{13}-\tfrac{p^*}{2}\right)
        =p_{13}+\tfrac{p^*}{2}
        \le \tfrac{3p_{13}}{2}
        = \tfrac{3(1-p_{11}-p_{12})}{2}
        \le \tfrac{1}{2}.
    \end{align}
    Additionally, agent $1$ envies agent $3$ with probability
    \begin{align}
    \textstyle
        \left(p_{12}-\frac{p^*}{2}\right)+\frac{p^*}{2}+\left(p_{13}-\frac{p^*}{2}\right)
        =p_{12}+p_{13}-\frac{p_{21}}{2}
        = \frac{p_{12}+p_{13}+p_{31}}{2}
        \le \frac{p_{12}+p_{13}+p_{11}}{2}
        = \frac{1}{2},
    \end{align}
    where the second equality holds by $p_{12}+p_{13} = 1 - p_{11} = p_{21}+p_{31}$.

    \item If $p_{22}=p^*$, then the decomposition becomes
    \begin{align}
    \tfrac{p^*}{2}A^{123}+
    (p_{11}-\tfrac{p^*}{2})A^{132}+
    (p_{33}-\tfrac{p^*}{2})A^{213}+
    (p_{31}-\tfrac{p^*}{2})A^{231}+
    (p_{13}-\tfrac{p^*}{2})A^{312}+
    \tfrac{p^*}{2}A^{321}.
    \end{align}
    In this decomposition, agent $1$ envies agent $2$ with probability
    \begin{align}
    \textstyle
        \left(p_{33}-\frac{p^*}{2}\right)+\left(p_{13}-\frac{p^*}{2}\right)+\frac{p^*}{2}
        =p_{21}+\frac{p_{22}}{2}
        \le \frac{p_{11}+p_{21}+p_{22}}{2}
        \le \frac{p_{11}+p_{21}+p_{31}}{2}
        = \frac{1}{2},
    \end{align}
    where the second inequality holds by $p_{22}=p^*\le p_{31}$.
    Additionally, agent $1$ envies agent $3$ with probability
    \begin{align}
    \textstyle
        \left(p_{31}-\frac{p^*}{2}\right)+\left(p_{13}-\frac{p^*}{2}\right)+\frac{p^*}{2}
        =p_{31}+p_{13}-\frac{p_{22}}{2}
        \le \frac{p_{31}+p_{11}+p_{13}+p_{33}}{2}-\frac{p_{22}}{2}
        = \frac{1}{2},
    \end{align}
    where the inequality holds by $p_{31}\le p_{11}$ and $p_{13}=1-(p_{11}+p_{12})\le 1-(p_{31}+p_{32})=p_{33}$, and the last equality holds by $p_{31}+p_{11} = 1-p_{21}$ and $p_{13}+p_{33} = 1 - p_{23}$.
    \item If $p_{23}=p^*$, we have $p_{13}=1-(p_{11}+p_{12})\le 1-(p_{21}+p_{22})=p_{23}=p^*$.
    Since $p^*$ is the minimum value, we obtain $p_{13} = p^*$.  
    Therefore, by the analysis of the case $p_{13} = p^*$, the probability that agent $1$ envies each of the other agents is at most $1/2$.
\end{itemize}
The remaining cases where $p_{31}=p^*$, $p_{32}=p^*$, or $p_{33}=p^*$ can be handled in the same way as the cases 
$p_{21}=p^*$, $p_{22}=p^*$, or $p_{23}=p^*$, respectively.
Therefore, the decomposition satisfies Dec-EF.
\end{proof}
\begin{comment}
We demonstrate that the following is a Dec-EF decomposition:
\begin{equation}
\begin{aligned}
P &= 
\frac{p_{11}}{2}\begin{bmatrix}1&0&0\\0&1&0\\0&0&1\end{bmatrix}+
\frac{p_{11}}{2}\begin{bmatrix}1&0&0\\0&0&1\\0&1&0\end{bmatrix}+
\left(p_{33}-\frac{p_{11}}{2}\right)\begin{bmatrix}0&1&0\\1&0&0\\0&0&1\end{bmatrix}\\
&\qquad+\left(p_{23}-\frac{p_{11}}{2}\right)\begin{bmatrix}0&1&0\\0&0&1\\1&0&0\end{bmatrix}+
\left(p_{32}-\frac{p_{11}}{2}\right)\begin{bmatrix}0&0&1\\1&0&0\\0&1&0\end{bmatrix}+
\left(p_{22}-\frac{p_{11}}{2}\right)\begin{bmatrix}0&0&1\\0&1&0\\1&0&0\end{bmatrix}.
\end{aligned}\label{eq:three-decomposition}
\end{equation}
Note that, by the assumption that $p_{11}$ is the smallest entry, each coefficient is nonnegative.
Moreover, decomposition~\eqref{eq:three-decomposition} is indeed a decomposition of $P$, since
\[
\begin{bmatrix}
p_{11} & p_{12} & p_{13} \\
p_{21} & p_{22} & p_{23} \\
p_{31} & p_{32} & p_{33}
\end{bmatrix}=
\begin{bmatrix}
p_{11} & p_{33}+p_{23}-p_{11} & p_{32}+p_{22}-p_{11} \\
p_{33}+p_{32}-p_{11} & p_{22} & p_{23} \\
p_{23}+p_{22}-p_{11} & p_{32} & p_{33}
\end{bmatrix},
\]
which follows from the fact that the entries in each row and each column of $P$ sum to~$1$.
\end{comment}

Next, we address the case where the agents have at most two distinct preferences.

\begin{theorem}
\label{thm:two-types}
For every instance such that the agents have at most two distinct preferences, every SD-EF assignment matrix is EF-decomposable.
\end{theorem}

\newcommand{\pos}{\mathrm{pos}}
\newcommand{\Sr}{\mathcal{S}_r}
\begin{proof}
Let us say that agents with the same preference are of the same type.
We will prove the theorem when there are exactly two agent types---if there is only one agent type, we can arbitrarily designate one agent to be of the second type and apply the same argument.

Consider any instance with two agent types, and assume that the first type and second type contain $r$ and $s$ agents, respectively, where $r,s \ge 1$.
Let $P$ be an SD-EF assignment matrix, where $N=\{1,2,\dots, n\}$ and $M=\{o_1, o_2, \dots, o_n\}$, and let $p_{ij}$ denote the probability that agent~$i$ receives object~$o_j$ according to~$P$.

Observe that $P_i=P_{i'}$ holds for any agents $i$ and $i'$ of the same type, because otherwise SD-EF would be violated. 
For each object $o_j$, let $a_j$ and $b_j$ denote the probability that $o_j$ is allocated to an agent of the first and second type, respectively. Then, for any $o_j$, we have $a_j+b_j=1$, and $p_{ij}=a_j/r$ for each first-type agent $i$ and $p_{i'j}=b_j/s$ for each second-type agent $i'$. 

We now construct a decomposition $\pi$ of $P$. Without loss of generality, we assume that the first $r$ rows correspond to the first-type agents.
Then, $P$ is a bistochastic matrix whose first $r$ rows are identical and whose remaining $s$ rows are identical.

\begin{enumerate}
\item 
Let $Q$ be any permutation matrix (i.e., a deterministic assignment) such that $q_{ij}=1$ only if $p_{ij}>0$ for any $i,j$. The existence of such $Q$ is guaranteed by the Birkhoff--von Neumann theorem. 
Divide $Q$ into two parts; $Q_1$ consists of the first $r$ rows and $Q_2$ consists of the remaining $s$ rows. 

\item Consider the set of matrices of size $r\times n$ obtained by cyclically shifting the rows in $Q_1$. There are $r$ such matrices, including $Q_1$. Also, let $Q_1^*$ be the matrix obtained from $Q_1$ by vertical reflection (i.e., reversing the order of rows), and consider the set of $r$ matrices obtained by cyclically shifting the rows in $Q_1^*$. Thus, we have $2r$ matrices of size $r \times n$ in total, each assigning the same set of objects to the first-type agents. Similarly, from $Q_2$, by vertical reflection and cyclic shifting, we define $2s$ matrices of size $s \times n$, each of which assigns the same set of objects to the second-type agents. 

\item By considering all combinations of these upper and lower matrices, we obtain a set of $2r \cdot 2s$ matrices of size $n \times n$, which we denote by $\mathcal{Q}$. Since these matrices are obtained from $Q$ by shifting and reflection within each type, and the first $r$ rows of~$P$ are identical and the remaining $s$ rows are also identical, any $Q' \in \mathcal{Q}$ satisfies $q'_{ij}=1$ only if $p_{ij}>0$.

\item As a result of shifting, the matrix $\sum_{Q' \in \mathcal{Q}} Q'$ has identical first $r$ rows and identical last $s$ rows. Moreover, as a result of vertical reflection (combined with shifting), for any agents $i,i'$ of the same type and any objects $o_j,o_{j'}$, 
the number of matrices in $\mathcal{Q}$ assigning $o_j$ to $i$ and $o_{j'}$ to $i'$ 
equals the number of those assigning $o_j$ to $i'$ and $o_{j'}$ to $i$ (i.e., either both are one or both are zero). This symmetry will ensure Dec-EF among agents of the same type in the resulting decomposition.

\item Let $\alpha$ be the maximum value such that $\alpha \sum_{Q' \in \mathcal{Q}} Q'$ does not exceed $P$ at any entry. Then, the matrix $P' \coloneqq P - \alpha \sum_{Q' \in \mathcal{Q}} Q'$ has more zero entries than $P$. More precisely, at least one type--object pair is eliminated. Since both $P$ and $\sum_{Q' \in \mathcal{Q}} Q'$ have the property that the first $r$ rows are identical and the remaining $s$ rows are identical, we see that there is an object $o_j$ such that $p_{ij} > p'_{ij} = 0$ for either every first-type agent $i$ or every second-type agent $i$.
We decompose $P$ as $\alpha \sum_{Q' \in \mathcal{Q}} Q'+P' = \sum_{Q' \in \mathcal{Q}}(\alpha \cdot Q') + P'$, where each $Q'$ is a permutation matrix.

\item We stop if the resultant matrix $P'$ is a zero matrix. If not, then with an appropriate scaling, $P'$ is a bistochastic matrix with the property that the first $r$ rows are identical and the remaining $s$ rows are identical. Therefore, we can apply the above operation again to decompose $P'$. Since at each step some type--object pair is eliminated, and the last step eliminates $n$ type--object pairs simultaneously, the number of repetitions is at most $n+1$.  
\end{enumerate}

This completes the construction\footnote{See \Cref{ex:two-types-decomposition} after this proof for a concrete example.} of the decomposition~$\pi$ of $P$.
From the construction, we see that the number of matrices with positive probabilities in $\pi$ is at most $4rs(n+1)$. 
We show that $\pi$ satisfies Dec-EF. 
The fact that no agent envies an agent of the same type with probability more than $1/2$ follows from the construction as explained above (in the fourth item).
Therefore, it remains to show that no agent envies an agent of a different type with probability more than $1/2$. 
By symmetry, it suffices to show that any first-type agent does not envy any second-type agent with probability more than $1/2$. Without loss of generality, we assume that the preferences of the first-type agents $i$ are $o_1 \succ_i o_2 \succ_i \cdots \succ_i o_n$.

Let $\Sr$ denote the family of subsets of $M$ of size $r$. For each $S\in \Sr$, let $a_S$ denote the probability that the set of objects allocated to the first-type agents coincides with $S$ under $\pi$. 
By definition, for each $o_j\in M$, we have $\sum_{S\in \Sr:\, o_j \in S} a_S = a_j$.
Take any first-type agent $i_1$ and any second-type agent $i_2$.
Note that when the first-type agents receive the set $S\in \Sr$, for each pair of objects $(o_j, o_{j'}) \in S \times (M \setminus S)$, the conditional probability that $i_1$ receives $o_j$ and $i_2$ receives $o_{j'}$ is $\tfrac{1}{rs}$ by the symmetry of $\pi$.
Moreover, $i_1$ envies $i_2$ if and only if $o_{j'}$ is one of the objects in $M \setminus S$ preceding $o_j$ according to the preference $o_1\succ_{i_1} o_2\succ_{i_1}\dots\succ_{i_1} o_n$. 
For each $o_j\in S$, there are $j-\pos(S,j)$ objects in $M \setminus S$ preceding $o_j$, where $\pos(S,j)$ is defined so that $o_j$ is the $\pos(S,j)$-th best object in $S$ for $i_1$.
Using these, the probability that $i_1$ envies $i_2$ can be written as
\[E(i_1, i_2) \coloneqq \sum_{S\in \Sr}a_S \left(\frac{1}{rs}\sum_{o_j\in S}(j-\pos(S,j))\right).\]
Our goal is to show that the value $E(i_1, i_2)$ is at most $1/2$, i.e., it holds that $\frac{1}{2}-E(i_1, i_2)\geq 0$.
Observe that $\sum_{S\in\Sr}a_S=1$ and, for each~$S$, we have $\sum_{o_j\in S}\pos(S,j)=(r+1)r/2$. Therefore, 
\begin{align}
\frac{1}{2}-E(i_1, i_2)&=\sum_{S\in \Sr}a_S \left(\frac{1}{2}-\frac{1}{rs}\sum_{o_j\in S}(j-\pos(S,j))\right) \\
&=\frac{1}{2rs}\sum_{S\in \Sr}a_S \left(r(s+r+1)-2\sum_{o_j\in S}j\right). \label{eq:2types-envy}
\end{align}
To show the nonnegativity of this expression, we use the following two claims.
\begin{claim}\label{claim:two-type-1}
For any $t\in \{1,2,\dots,n\}$, we have $\sum_{j=1}^t (a_j-a_{n-j+1})\geq 0$.
\end{claim}
\begin{claimproof}
Because $P$ is SD-EF, the random allocation of a first-type agent stochastically dominates that of a second-type agent with respect to the former agent's preference. 
Hence, for each $\tau\in \{1,2,\dots,n\}$, it holds that $\sum_{j=1}^\tau a_j/r\geq \sum_{j=1}^\tau b_j/s$.
Moreover, we have $\sum_{j=1}^n a_j/r=1$. By substituting $b_j=1-a_j$, these conditions can be rewritten as 
\[\sum_{j=1}^\tau a_j\geq \frac{r}{r+s}\cdot \tau\quad (\tau=1,2,\dots, n-1),\qquad
\sum_{j=1}^n a_j= \frac{r}{r+s}\cdot n.\]
For each $t\in\{1,2,\dots,n\}$, adding the inequalities on the left with $\tau=t$ and $\tau=n-t$, and subtracting the equality on the right, we get
\[0=\frac{r}{r+s}\cdot t +\frac{r}{r+s}\cdot (n-t) -\frac{r}{r+s}\cdot n\leq \sum_{j=1}^t a_j+\sum_{j=1}^{n-t} a_j-\sum_{j=1}^n a_j
=\sum_{j=1}^t(a_j-a_{n-j+1}),\]
which completes the proof of the claim.
\end{claimproof}
\begin{claim}\label{claim:two-type-2}
$\sum_{t=1}^n \sum_{j=1}^t (a_j-a_{n-j+1})=\sum_{S\in \Sr}a_S \left(r(s+r+1)-2\sum_{o_j\in S}j\right)$.
\end{claim}
\begin{claimproof}
Recall that we have $\sum_{S\in \Sr:\, o_j\in S}a_S=a_j$ for every $o_j\in M$. By substituting this, we obtain the following chain of equalities:
\begin{align*}
\sum_{t=1}^n \sum_{j=1}^t (a_j-a_{n-j+1})
&=\sum_{t=1}^n \sum_{j=1}^t \left(\sum_{S\in \Sr:\, o_j\in S}a_S - \sum_{S\in \Sr:\, o_{n-j+1}\in S}a_S\right)\\    
&=\sum_{S\in \Sr} \left[ a_S \sum_{t=1}^n \left(\sum_{j:\, o_j\in S,\, j\leq t} 1 - \sum_{j:\, o_{n-j+1}\in S,\, j\leq t} 1\right)\right]\\    
&=\sum_{S\in \Sr} \left[ a_S \sum_{o_j\in S} \left(\sum_{t=j}^n 1 - \sum_{t=n-j+1}^n 1\right)\right]\\
&=\sum_{S\in \Sr} \left[ a_S \sum_{o_j\in S} \left[(n-j+1) - (n-(n-j+1)+1)\right]\right]\\
&=\sum_{S\in \Sr} \left[ a_S \sum_{o_j\in S} (n+1-2j) \right] =\sum_{S\in \Sr}a_S \left(r(n+1)-2\sum_{o_j\in S}j\right).
\end{align*}
As $n=s+r$, we obtain the desired equality.
\end{claimproof}
We now complete the proof of Theorem~\ref{thm:two-types}. By Claim~\ref{claim:two-type-2} and Equation~\eqref{eq:2types-envy}, we have
\begin{align*}
\frac{1}{2}-E(i_1, i_2)=\frac{1}{2rs}\sum_{t=1}^n \sum_{j=1}^t (a_j-a_{n-j+1}),
\end{align*}
and Claim~\ref{claim:two-type-1} implies that $\sum_{j=1}^t (a_j-a_{n-j+1})\geq 0$ for each $t\in\{1,2,\dots, n\}$. Thus, we obtain the inequality $\frac{1}{2}-E(i_1, i_2)\geq 0$, as desired. 
\end{proof}

We provide an example illustrating the decomposition procedure from \Cref{thm:two-types}.

\begin{example}
\label{ex:two-types-decomposition}
Consider an instance with $n = 4$ agents such that the first three agents have the preference $o_1\succ_i o_2\succ_i o_3\succ_i o_4$ for $i\in\{1,2,3\}$, while the fourth agent has the preference $o_2\succ_4 o_3\succ_4 o_4\succ_4 o_1$.
The assignment matrix
\[  
P = \begin{bmatrix}
1/3 & 1/6 & 1/4 & 1/4 \\
1/3 & 1/6 & 1/4 & 1/4 \\
1/3 & 1/6 & 1/4 & 1/4 \\
0 & 1/2 & 1/4 & 1/4 \\
\end{bmatrix}
\]
satisfies SD-EF (in fact, it is the output of PS), where the rows correspond to agents $1,2,3,4$ and the columns correspond to objects $o_1,o_2,o_3,o_4$, respectively. 

To apply the decomposition procedure from \Cref{thm:two-types} on~$P$, we first identify an appropriate permutation matrix~$Q$.
One such permutation matrix is
\[
Q = \begin{bmatrix}
1 & 0 & 0 & 0 \\
0 & 0 & 1 & 0 \\
0 & 0 & 0 & 1 \\
0 & 1 & 0 & 0 \\
\end{bmatrix}.
\]
With this~$Q$, the set~$\mathcal{Q}$ consists of $(2\cdot 3)\cdot (2\cdot 1) = 12$ permutation matrices obtained from $Q$ by shifting and reflection within each type of agents, and
\[
\sum_{Q'\in\mathcal{Q}} Q' = \begin{bmatrix}
4 & 0 & 4 & 4 \\
4 & 0 & 4 & 4 \\
4 & 0 & 4 & 4 \\
0 & 12 & 0 & 0 \\
\end{bmatrix}.
\]
The largest $\alpha$ such that $\alpha \sum_{Q' \in \mathcal{Q}} Q'$ does not exceed $P$ at any entry is $\alpha = 1/24$, and the resultant matrix $P' \coloneqq P - \alpha \sum_{Q' \in \mathcal{Q}} Q'$ is
\[
P' = \begin{bmatrix}
1/6 & 1/6 & 1/12 & 1/12 \\
1/6 & 1/6 & 1/12 & 1/12 \\
1/6 & 1/6 & 1/12 & 1/12 \\
0 & 0 & 1/4 & 1/4 \\
\end{bmatrix}.
\]

As $P'$ is not a zero matrix, we can continue by taking\footnote{The procedure from \Cref{thm:two-types} rescales $P'$ to be a bistochastic matrix.
However, for the sake of convenience, we will not do that here.}
\[
Q = \begin{bmatrix}
1 & 0 & 0 & 0 \\
0 & 1 & 0 & 0 \\
0 & 0 & 1 & 0 \\
0 & 0 & 0 & 1 \\
\end{bmatrix};
\qquad
\sum_{Q'\in\mathcal{Q}} Q' = \begin{bmatrix}
4 & 4 & 4 & 0 \\
4 & 4 & 4 & 0 \\
4 & 4 & 4 & 0 \\
0 & 0 & 0 & 12 \\
\end{bmatrix}.
\]
The corresponding $\alpha$ is $\alpha = 1/48$, which results in
\[
P' = \begin{bmatrix}
1/12 & 1/12 & 0 & 1/12 \\
1/12 & 1/12 & 0 & 1/12 \\
1/12 & 1/12 & 0 & 1/12 \\
0 & 0 & 1/4 & 0 \\
\end{bmatrix}.
\]

Lastly, we take
\[
Q = \begin{bmatrix}
1 & 0 & 0 & 0 \\
0 & 1 & 0 & 0 \\
0 & 0 & 0 & 1 \\
0 & 0 & 1 & 0 \\
\end{bmatrix};
\qquad
\sum_{Q'\in\mathcal{Q}} Q' = \begin{bmatrix}
4 & 4 & 0 & 4 \\
4 & 4 & 0 & 4 \\
4 & 4 & 0 & 4 \\
0 & 0 & 12 & 0 \\
\end{bmatrix}.
\]
The corresponding $\alpha$ is again $\alpha = 1/48$, which results in
\[
P' = \begin{bmatrix}
0 & 0 & 0 & 0 \\
0 & 0 & 0 & 0 \\
0 & 0 & 0 & 0 \\
0 & 0 & 0 & 0 \\
\end{bmatrix},
\]
and the procedure terminates.
In the final decomposition:
\begin{itemize}
\item each of the six deterministic assignments where agent~$4$ receives $o_2$ has a probability of $1/12$;
\item each of the six deterministic assignments where agent~$4$ receives $o_4$ has a probability of $1/24$;
\item each of the six deterministic assignments where agent~$4$ receives $o_3$ has a probability of $1/24$.
\end{itemize}
Under this decomposition, the probability that agent~$1$ envies agent~$4$ is $5/12 < 1/2$, and the probability that agent~$4$ envies agent~$1$ is $1/4 < 1/2$.
The same holds when agent~$1$ is replaced by agent~$2$ or $3$.
\end{example}

Since the assignment matrix output by PS satisfies SD-EF \citep{BogomolnaiaMo01}, \Cref{thm:three-agents,thm:two-types} imply the following corollary.
\begin{corollary}
For every instance with at most three agents, the assignment matrix output by PS is EF-decomposable.
The same holds for every instance where agents have at most two distinct preferences.
\end{corollary}
It remains an intriguing question whether every SD-EF assignment matrix (or every matrix output by PS) is EF-decomposable in general.
This is true for all instances that we have checked using a computer search; in particular, we verified that for every instance with four agents, the matrix output by PS is EF-decomposable.
Unfortunately, the techniques that we developed for proving \Cref{thm:three-agents,thm:two-types} are not sufficient to answer the question.
Nevertheless, we shall derive upper bounds on the possible envy between agents in the next two propositions.

\begin{proposition}
\label{prop:SD-EF-upper}
Let $n\ge 2$ be a positive integer.
Consider any instance with $n$~agents, and let $i,i'\in N$ be any pair of agents.
For every SD-EF assignment matrix~$P$ and every decomposition~$\pi$ of~$P$, the probability that agent~$i$ envies agent~$i'$ according to~$\pi$ is at most $\frac{n-1}{n}$.
\end{proposition}

\begin{proof}
Since $P$ satisfies SD-EF, $P_i$ stochastically dominates $P_{i''}$ with respect to $\succ_i$ for every $i''\neq i$.
In particular, the probability that agent~$i$ receives her most preferred object under~$P$ must be at least as high as the probability that agent~$i$ receives this object under~$P_{i''}$ for every $i''\neq i$. 
Since there are $n$~agents in total, this means that the probability that agent~$i$ receives her most preferred object under~$P$ is at least $1/n$.
Hence, the same is true for the decomposition~$\pi$.
Whenever agent~$i$ receives her most preferred object, she does not envy any other agent.
Therefore, the probability that agent~$i$ receives her most preferred object is at most the probability that she does not envy any other agent; in other words, the probability that agent~$i$ envies some other agent is at most the probability that she does not receive her most preferred object.
It follows that the probability that agent~$i$ envies agent~$i'$ according to~$\pi$ is at most $1 - \frac{1}{n} = \frac{n-1}{n}$.
\end{proof}

The bound $\frac{n-1}{n}$ in \Cref{prop:SD-EF-upper} is tight: for the instance in \Cref{ex:cyclic} (generalized to $n$~agents) and the uniform assignment matrix, under the cyclic decomposition, agent~$2$ envies agent~$1$ with probability exactly $\frac{n-1}{n}$.
However, for the assignment matrix produced by PS, if we choose the decomposition more carefully, we can slightly improve the envy bound.
Recall that PS works by letting agents ``eat'' their most preferred available object simultaneously at the same speed---the fraction of an object that an agent has eaten is interpreted as the probability with which the agent receives that object in the resulting assignment matrix.

\begin{proposition}
Let $n\ge 3$ be a positive integer, and consider any instance with $n$~agents.    
For the assignment matrix output by PS, there exists a decomposition~$\pi$ such that for every pair of agents $i,i'\in N$, the probability that agent~$i$ envies agent~$i'$ according to $\pi$ is at most $\frac{n-2}{n-1}$.
\end{proposition}

\begin{proof}
If all agents have the same preference, PS produces the uniform assignment matrix.
Under the uniform decomposition of this matrix (where each of the $n!$ deterministic assignments receives probability $\frac{1}{n!}$), for every pair of agents $i,i'\in N$, agent~$i$ envies agent~$i'$ with probability $\frac{1}{2} \le \frac{n-2}{n-1}$.

Assume now that not all agents have the same preference, and fix a pair of agents $i,i'\in N$.
There must exist $j\in\{1,\dots,n-1\}$ such that all agents agree on the $j'$-th most preferred object for every $1\le j' < j$, but not all of them agree on the $j$-th most preferred object.\footnote{If $j = 1$, then not all agents agree on their most preferred object.}
Let $o_1,\dots,o_{j-1}$ denote the $j-1$ most preferred objects of all agents in this order, and let $o_j$ denote the $j$-th most preferred object of agent~$i$.
According to the assignment matrix output by PS, each agent receives each of $o_1,\dots,o_{j-1}$ with probability $\frac{1}{n}$; moreover, agent~$i$ receives $o_j$ with probability at least $\frac{1}{n-1}$.
Hence, the same is true for any decomposition of the assignment matrix.
When agent~$i$ receives one of $o_1,\dots,o_j$ and agent~$i'$ does not receive one of $o_1,\dots,o_{j-1}$, agent~$i$ will not envy agent~$i'$.
Since the probability that agent~$i$ receives one of $o_1,\dots,o_j$ exceeds the probability that agent~$i'$ receives one of $o_1,\dots,o_{j-1}$ by at least $\frac{1}{n-1}$, we have that with probability at least $\frac{1}{n-1}$, agent~$i$ does not envy agent~$i'$.
Therefore, the probability that agent~$i$ envies agent~$i'$ is at most $1 - \frac{1}{n-1} = \frac{n-2}{n-1}$.
\end{proof}

In RP, the permutation of agents is chosen uniformly at random.
More generally, observe from the proof of \Cref{prop:RP} that if a random assignment can be implemented by a ``reversal-symmetric'' probability distribution over permutations of agents, it satisfies Dec-EF.
Here, a distribution is called \emph{reversal-symmetric} if the probability of any permutation of agents is equal to that of the reverse permutation.
Therefore, a natural strategy to prove that PS satisfies EF-decomposability would be to demonstrate that its outcome can always be implemented by a reversal-symmetric distribution over serial dictatorships.
However, the following example illustrates that this strategy fails: there exists an assignment matrix produced by PS that cannot be implemented by any reversal-symmetric distribution (but nevertheless admits a Dec-EF decomposition).

\begin{example}\label{ex:symmetric}
Consider an instance with $n = 4$ agents whose preferences are 
$a\succ_1 b\succ_1 c\succ_1 d$,\, 
$a\succ_2 b\succ_2 c\succ_2 d$,\, 
$a\succ_3 c\succ_3 b\succ_3 d$, and 
$c\succ_4 a\succ_4 b\succ_4 d$.
Let $P$ be the assignment matrix produced by PS, i.e.,
\[ P = 
\begin{bmatrix}
1/3 & 5/12 & 0   & 1/4\\
1/3 & 5/12 & 0   & 1/4\\
1/3 & 1/12 & 1/3 & 1/4\\
0   & 1/12 & 2/3 & 1/4
\end{bmatrix},
\]
where the rows correspond to agents $1,2,3,4$ and the columns correspond to objects $a,b,c,d$, respectively.  
The matrix $P$ is EF-decomposable---for example, the decomposition
\begin{align}
\pi=\frac{1}{12}\cdot\begin{bmatrix}1 & 0 & 0 & 0\\0 & 1 & 0 & 0\\0 & 0 & 1 & 0\\0 & 0 & 0 & 1\end{bmatrix}+
\frac{1}{12}\cdot\begin{bmatrix}1 & 0 & 0 & 0\\0 & 0 & 0 & 1\\0 & 1 & 0 & 0\\0 & 0 & 1 & 0\end{bmatrix}+
\frac{1}{12}\cdot\begin{bmatrix}0 & 1 & 0 & 0\\1 & 0 & 0 & 0\\0 & 0 & 0 & 1\\0 & 0 & 1 & 0\end{bmatrix}+
\frac{1}{12}\cdot\begin{bmatrix}0 & 0 & 0 & 1\\1 & 0 & 0 & 0\\0 & 0 & 1 & 0\\0 & 1 & 0 & 0\end{bmatrix}\\
+\frac{1}{6}\cdot\begin{bmatrix}1 & 0 & 0 & 0\\0 & 1 & 0 & 0\\0 & 0 & 0 & 1\\0 & 0 & 1 & 0\end{bmatrix}+
\frac{1}{6}\cdot\begin{bmatrix}0 & 1 & 0 & 0\\1 & 0 & 0 & 0\\0 & 0 & 1 & 0\\0 & 0 & 0 & 1\end{bmatrix}+
\frac{1}{6}\cdot\begin{bmatrix}0 & 1 & 0 & 0\\0 & 0 & 0 & 1\\1 & 0 & 0 & 0\\0 & 0 & 1 & 0\end{bmatrix}+
\frac{1}{6}\cdot\begin{bmatrix}0 & 0 & 0 & 1\\0 & 1 & 0 & 0\\1 & 0 & 0 & 0\\0 & 0 & 1 & 0\end{bmatrix}
\end{align}
satisfies Dec-EF.

Next, we show that despite satisfying Dec-EF, the assignment matrix $P$ cannot be implemented by any reversal-symmetric distribution over permutations of agents.
Table~\ref{tab:symmetric} lists the serial dictatorship outcomes for all permutations grouped by reversal-symmetry.

\begin{table}[t]
  \centering
  \caption{Resulting deterministic assignment for each permutation of agents in \Cref{ex:symmetric}}
  \label{tab:symmetric}
  \begin{tabular}{c|c|c}
    \toprule
    Permutation & Assignment                     & Probability\\
    \midrule
    $1, 2, 3, 4$  & $(1{:}a, 2{:}b, 3{:}c, 4{:}d)$ & $p_1$      \\
    $4, 3, 2, 1$  & $(1{:}d, 2{:}b, 3{:}a, 4{:}c)$ & $p_1$      \\
    \hline
    $2, 1, 3, 4$  & $(1{:}b, 2{:}a, 3{:}c, 4{:}d)$ & $p_2$      \\
    $4, 3, 1, 2$  & $(1{:}b, 2{:}d, 3{:}a, 4{:}c)$ & $p_2$      \\
    \hline
    $1, 2, 4, 3$  & $(1{:}a, 2{:}b, 3{:}d, 4{:}c)$ & $p_3$      \\
    $3, 4, 2, 1$  & $(1{:}d, 2{:}b, 3{:}a, 4{:}c)$ & $p_3$      \\
    \hline
    $2, 1, 4, 3$  & $(1{:}b, 2{:}a, 3{:}d, 4{:}c)$ & $p_4$      \\
    $3, 4, 1, 2$  & $(1{:}b, 2{:}d, 3{:}a, 4{:}c)$ & $p_4$      \\
    \hline
    $1, 3, 2, 4$  & $(1{:}a, 2{:}b, 3{:}c, 4{:}d)$ & $p_5$      \\
    $4, 2, 3, 1$  & $(1{:}d, 2{:}a, 3{:}b, 4{:}c)$ & $p_5$      \\
    \hline
    $2, 3, 1, 4$  & $(1{:}b, 2{:}a, 3{:}c, 4{:}d)$ & $p_6$      \\
    $4, 1, 3, 2$  & $(1{:}a, 2{:}d, 3{:}b, 4{:}c)$ & $p_6$      \\
    \hline    
    $1, 3, 4, 2$  & $(1{:}a, 2{:}d, 3{:}c, 4{:}b)$ & $p_7$      \\
    $2, 4, 3, 1$  & $(1{:}d, 2{:}a, 3{:}b, 4{:}c)$ & $p_7$      \\
    \hline
    $2, 3, 4, 1$  & $(1{:}d, 2{:}a, 3{:}c, 4{:}b)$ & $p_8$      \\
    $1, 4, 3, 2$  & $(1{:}a, 2{:}d, 3{:}b, 4{:}c)$ & $p_8$      \\
    \hline
    $1, 4, 2, 3$  & $(1{:}a, 2{:}b, 3{:}d, 4{:}c)$ & $p_9$      \\
    $3, 2, 4, 1$  & $(1{:}d, 2{:}b, 3{:}a, 4{:}c)$ & $p_9$      \\
    \hline
    $2, 4, 1, 3$  & $(1{:}b, 2{:}a, 3{:}d, 4{:}c)$ & $p_{10}$   \\
    $3, 1, 4, 2$  & $(1{:}b, 2{:}d, 3{:}a, 4{:}c)$ & $p_{10}$   \\
    \hline
    $3, 1, 2, 4$  & $(1{:}b, 2{:}c, 3{:}a, 4{:}d)$ & $p_{11}$   \\
    $4, 2, 1, 3$  & $(1{:}b, 2{:}a, 3{:}d, 4{:}c)$ & $p_{11}$   \\
    \hline
    $3, 2, 1, 4$  & $(1{:}c, 2{:}b, 3{:}a, 4{:}d)$ & $p_{12}$   \\
    $4, 1, 2, 3$  & $(1{:}a, 2{:}b, 3{:}d, 4{:}c)$ & $p_{12}$   \\
    \bottomrule
  \end{tabular}
\end{table}

For $i\in\{1,2,\dots,12\}$, let $p_i$ denote the probability assigned to each permutation in the $i$-th group. For $P$ to be implementable by a reversal-symmetric distribution, the following conditions must hold:
\begin{itemize}
    \item[(i)] Since the total probability is $1$, we have $\sum_{i=1}^{12}2p_i=1$.
    \item[(ii)] Since agent $3$ receives $a$ with probability $1/3$, we have $p_1+p_2+p_3+p_4+p_9+p_{10}+p_{11}+p_{12}=1/3$.
    \item[(iii)] Since agent $3$ receives $b$ with probability $1/12$, we have $p_5+p_6+p_7+p_8=1/12$.
\end{itemize}
By (ii) and (iii) we obtain $\sum_{i=1}^{12}p_i=1/3+1/12=5/12$, whereas by (i) we have $\sum_{i=1}^{12}p_i=1/2$, a contradiction.
Hence, the assignment matrix $P$ cannot be implemented by any reversal-symmetric distribution over permutations.
\end{example}

\section{Relations to Other Notions}

We end the paper by investigating the relations between Dec-EF and some existing notions in the random assignment literature.

An assignment matrix $P$ is said to satisfy \emph{equal treatment of equals} if $P_i = P_{i'}$ for any agents $i,i'\in N$ such that ${\succ_i} = {\succ_{i'}}$.
\Cref{ex:weak-SD-EF-Dec-EF} shows that an assignment matrix that satisfies equal treatment of equals is not necessarily EF-decomposable.
The following example demonstrates that the converse also fails to hold: EF-decomposability does not imply equal treatment of equals.

\begin{example}
Consider an instance with $n = 3$ agents whose preferences are $a\succ_1 b\succ_1 c$,\, $a\succ_2 b\succ_2 c$, and $c\succ_3 a\succ_3 b$.
Let $P$ be the assignment matrix
\[ P = 
\begin{bmatrix}
0.4 & 0.4 & 0.2 \\
0.5 & 0.4 & 0.1 \\
0.1 & 0.2 & 0.7
\end{bmatrix},
\]
where the rows correspond to agents $1,2,3$ and the columns correspond to objects $a,b,c$, respectively.    
The matrix $P$ is EF-decomposable---for example, the decomposition
\[ \pi = 0.4 \cdot 
\begin{bmatrix}
1 & 0 & 0 \\
0 & 1 & 0 \\
0 & 0 & 1
\end{bmatrix} 
+ 0.3 \cdot 
\begin{bmatrix}
0 & 1 & 0 \\
1 & 0 & 0 \\
0 & 0 & 1
\end{bmatrix}
+ 0.1 \cdot 
\begin{bmatrix}
0 & 1 & 0 \\
0 & 0 & 1 \\
1 & 0 & 0
\end{bmatrix}
+ 0.2 \cdot 
\begin{bmatrix}
0 & 0 & 1 \\
1 & 0 & 0 \\
0 & 1 & 0
\end{bmatrix}
\]
satisfies Dec-EF.
However, $P$ does not satisfy equal treatment of equals, since $P_1 \ne P_2$ even though ${\succ_1} = {\succ_2}$.
\end{example}

Next, an assignment matrix~$P$ is said to satisfy \emph{SD-efficiency} if there does not exist another assignment matrix~$P'$ such that $P'_i$ stochastically dominates $P_i$ for every $i\in N$ and $P_i \neq P'_i$ for at least one $i\in N$.
The following example shows that SD-efficiency, even when combined with weak SD-EF and equal treatment of equals, does not imply EF-decomposability. 

\begin{example}
Consider an instance with $n = 4$ agents whose preferences are $a\succ_1 b\succ_1 c\succ_1 d$ and $a\succ_i b\succ_i d\succ_i c$ for $i\in\{2,3,4\}$.
Let $P$ be the assignment matrix
\[ P = 
\begin{bmatrix}
0   & 0   & 1 & 0   \\
1/3 & 1/3 & 0 & 1/3 \\
1/3 & 1/3 & 0 & 1/3 \\
1/3 & 1/3 & 0 & 1/3 \\
\end{bmatrix},
\]
where the rows correspond to agents $1,2,3,4$ and the columns correspond to objects $a,b,c,d$, respectively.
One can verify that $P$ simultaneously satisfies SD-efficiency, weak SD-EF, and equal treatment of equals.
However, for any decomposition of $P$, agent~$1$ envies agent~$2$ with probability~$2/3$, and thus no decomposition of $P$ satisfies Dec-EF.
\end{example}

Another related notion is ex-post efficiency, which is weaker than SD-efficiency.
A deterministic assignment is said to be \emph{Pareto optimal} if no other deterministic assignment makes no agent worse off and at least one agent better off.
A random assignment~$\pi$ is called \emph{ex-post efficient} if every deterministic assignment with a positive probability under $\pi$ is Pareto optimal.
The relation between SD-efficiency and ex-post efficiency is similar to that between SD-EF and Dec-EF, in the sense that SD-efficiency and SD-EF are solely properties of an assignment matrix whereas ex-post efficiency and Dec-EF depend on the decomposition.

Since RP satisfies EF-decomposability (\Cref{prop:RP}) along with SD-strategyproofness\footnote{See the work of \citet{BogomolnaiaMo01} for the definition.}  and weak SD-EF \citep{BogomolnaiaMo01}, these properties are compatible with one another.
Also, if it holds that PS satisfies EF-decomposability, the notion would be compatible with SD-efficiency, SD-EF, and weak SD-strategyproofness.

\subsection*{Acknowledgments}

This work was partially supported by 
JST ERATO Grant Number JPMJER2301, 
by JST CRONOS Grant Number JPMJCS24K2, 
by JSPS KAKENHI Grant Numbers JP25K00137,
JP21K17708, 
and JP21H03397, 
Japan, 
by Value Exchange Engineering, a joint research project between Mercari R4D Lab and RIISE (Research Institute for an Inclusive Society through Engineering),
by the Singapore Ministry of Education under grant number MOE-T2EP20221-0001, and by an NUS Start-up Grant.  
We thank the anonymous reviewers for valuable feedback, and Ayumi Igarashi and Naoyuki Kamiyama for helpful discussions at the early stage.

\bibliographystyle{plainnat}
\bibliography{main}

\end{document}